%% file: main.tex
\documentclass[a4paper]{article}
\usepackage{amsthm}
\newtheorem{theorem}{Theorem}[section]

\newcommand{\xauthor}[3]{\begin{tabular}{c}#1\\#2\\\texttt{#3}\end{tabular}}

\newif\ifREVEDIT%
\newif\ifExamples%

\input{macros}

\begin{document}


\title{%
	An Algebraic Approach to Weighted
	Answer~Set~Programming
}
\author{%
	\xauthor{Francisco Coelho}{NOVA-LINCS, University of \'Evora}{fc@uevora.pt}   \and %
	\xauthor{Bruno Dinis}{CIMA, University of \'Evora}{bruno.dinis@uevora.pt}        \and %
	\xauthor{Dietmar Seipel}{Universit\"at W\"urzburg}{dietmar.seipel@uni-wuerzburg.de}      \and %
	\xauthor{Salvador Abreu}{NOVA-LINCS, University of \'Evora}{spa@uevora.pt}     %
}

\maketitle

\begin{abstract}
	\Aclp{LP}, more specifically, \aclp*{ASP}, can be annotated with probabilities on facts to express uncertainty.
	We address the 	problem of propagating weight annotations on facts (\eg\ probabilities) of an \acl*{ASP} to its \aclp*{SM}, and from there to events (defined as sets of atoms) in a dataset over the program's domain. %

	We propose a novel approach which is algebraic in the sense that it relies on an equivalence relation over the set of events.
	Uncertainty is then described as polynomial expressions over variables.
	We propagate the weight function in the space of models and events, rather than doing so within the syntax of the program. As evidence that our approach is sound, we show that certain facts behave as expected.
	Our approach allows us to investigate weight annotated programs and to determine how suitable a given one is for modeling a given dataset containing events.
\end{abstract}


\section{Introduction}
\label{sec:introduction}

\topicnote{logic+uncertainty}\noindent Using \iac{LP} to model and reason over a real world scenario is often difficult because of uncertainty underlying the problem being worked on.
Classic \acp{LP} represent knowledge in precise and complete terms, which turns out to be problematic when the scenario is characterized by stochastic or observability factors.\sidenote{insert simple, nice example}
%
%
We aim to explore how \acp{ASP} plus weight annotated facts can lead to useful characterizations for this class of problems.

To setup a working framework, we make the following assumption:
\begin{assumption}[System Representation, Data and Propagation]\label{asp:system.representation.data}

	\phantom{~}
	Consider a \emph{system} whose states are \emph{partially observable} (\ie, observations can miss some state information) or \emph{stochastic} (\ie\ observed values are affected by random noise).
	We assume that knowledge about such system features a formal specification including \emph{weighted facts} and empirical \emph{data} such that:
	\begin{description}\tight
		\item[Representation.]
		\topicnote{system = representation, \acs*{SM}=state}
		The system has a formal \emph{representation}\footnote{We use `representation' instead of `model' to avoid confusion with the \emph{stable models} of \aclp{ASP}.} in the form of a certain \acl{LP};
		The program's \aclp*{SM} correspond one-to-one with the system states.
		\item[Data.] \emph{Data} is a set of observations; a single \emph{observation} (of the system states) results from a set of (boolean) \emph{sensors}.
		\item[Propagation.]
		\topicnote{weight propagation}
		The \emph{weights} in \emph{facts} are \emph{propagated} to the \aclp{SM} of the representation.
	\end{description}
\end{assumption}

\topicnote{parameter estimation}
In this setting, data from observations can be used to estimate some parameteres used in the propagation process and, more importantly, to address the question of `\emph{How accurate is the representation of the system?}'.%

\topicnote{other \acs*{PLP}}%
Other \ac{PLP} systems such as \texttt{Problog} \cite{de2007problog}, \texttt{P-log} \cite{baral2009probabilistic} or \lpmln \cite{lee2016weighted}, in line with \cite{kifer1992theory}, derive a probability distribution of the \aclp{SM} from the \textit{syntax} of an annotated logic program.  The more expressive of these systems, according to \cite{lee2017lpmln}, is \lpmln, that can embed the \ac{MAP} estimation of the other systems.
\topicnote{probability from syntax}%
Since in these \ac{PLP} systems the probability distribution results from the program's syntax, these systems are limited to handle \emph{a priori} information and are unable to support \emph{posterior} data. Furthermore, these distributions are limited to the \acp{SM}, while we \emph{extend their domain to any event} (\ie\ any set of atoms). 

However, the key feature that we aim to address is concerned with the inherent uncertainty of the \acp{SM} in \aclp{LP} such as $\weightfact{a}{0.3}, b \vee c \clause a$. Intuitively, this program entails three \acp{SM}: $ac$, $ab$ and $\neg a$. We can, for example, assign the probability of $\neg a$ as $0.7$ but what about $ab$ and $ac$? Systems like \lpmln\ assign the same probability to these \acp{SM} \cite{lee2017lpmln,cozman2020joy}.

We question the underlying assumption of such assignments and propose a method where the distribution that results from the representation includes parameters (\ie\ variables in the algebraic sense: symbols for unknown quantities) that express the lack of \emph{prior} information concerning cases as above. The values of these parameters can be estimated \emph{a posteriori}, in the presence of data.

To frame this setting we extend our assumptions:
\begin{assumption}[Sensor Representation and Events]\label{asp:sensor.representation.events}
	\begin{description}\tight
		\item[Sensors.] \topicnote{sensor = atom}The \emph{sensors} of the system's states are associated to some atoms in the representation; 
		\item[Events.] \topicnote{event = atoms}An \emph{event} is a set of atoms from the representation.\sidenote{reconsider: is this a definition?}
	\end{description}
\end{assumption}

\topicnote{strong + week negation}More specifically, following the terminology set in \cite{calimeri2020aspcore}, a sensor $\sigma_a$ can `activate' the associated (\emph{classical}, \emph{strong}) \emph{atom} $a$ whereas no activation is represented by the (\emph{default}, \emph{weak} or \emph{negation-as-failure (naf)}) \emph{literal} $\naf a$. The same applies to a sensor $\sigma_{\neg a}$ associated to the (classically) negated atom $\neg a$.

\topicnote{hidden + stochastic}\Cref{asp:sensor.representation.events} enables a straightforward representation of \emph{hidden} parts of a system as well as faulty (\emph{stochastic}) sensors. For example, in the event
\begin{equation}
	\left\{ a, \neg a, \naf b, \naf \neg b, \naf c, \neg c \right\}\label{eq:example.event.long}	
\end{equation}
both $a$ and $\neg a$ are activated (suggesting \eg\ a problem in the associated sensors), $b$ is not observed (\ie\ hidden), and $\naf c, \neg c$ reports the (consistent) activation of $\neg c$ and no activation of $c$.
\topicnote{event $>$ observation} 
While every observation is an event, since some events can contain atoms not associated to sensors, not all events are observations. Furthermore, some events coincide with a \acl{SM} but others don't; an event may not uniquely determine a state of the system --- how to associate events to \aclp*{SM} and, thus, to system states, is addressed in \cref{sec:propagating.weights}.

\topicnote{event short notation}If we (i) omit the naf-literals; (ii) use $\co{x}$ to denote the classical $\neg x$; (iii) and use expressions like $ab$ to denote sets of literals such as $\set{a, b}$, then the event in (\ref{eq:example.event.long}) can be shortened to the equivalent form
\begin{equation}	
	a\co{a}\co{c}.\label{eq:example.event.short}
\end{equation}
Here we follow the convention, from \cite{gelfond1988stable}, of denoting a model by the set of true atoms, stressing that `\emph{falsehood}' results only from the default negation \ie\ $\naf a$ (\ie\ `\texttt{not a}' or `\texttt{\textbackslash+ a}' in logic languages).
More precisely, a model denotation can contain atoms such as $a$ or $\co{b}$ but not literals $\naf a, \naf \co{b}$. 
Our choice to represent sensor input using both positive and negative atoms is based on the following points: 
(i) it can be the case that there are two different sensors, $\sigma_a$ for the `positive' values and $\sigma_{\co{a}}$ for the `negative' ones;
(ii) the case where a single sensor $\sigma_{\hat{a}}$ always yields either `positive' or `negative' values can be represented by the rule $a \vee \neg a$; also, 
(iii) a closed-world assumption, where absence of sensor output means (classical) negation, can be represented by the rules $\co{a} \clause \naf a, a \clause \naf \co{a}$. 

\topicnote{weights + probabilities}%
Like in \lpmln, we annotate facts (\ie\ atoms) with weights \cite{lee2017lpmln} instead of probabilities, that result from normalization of the former.
\topicnote{propagation}
By \emph{propagation} we mean the use of those weights to define a `weight' function on the \aclp{SM} and then extended it to all the events in the program domain.
\topicnote{\acp{SM} + non-unique/non-deterministic}%
The step from facts to \acp{SM} is non-deterministic in the sense that a given set of facts may entail zero, one or more \acp{SM} (see \cref{ex:fruitful}).
This is a well-known situation, explained in \cref{sec:syntax.and.semantics} and also in \cite{verreet2022inference,pajunen2021solution,cozman2020joy,baral2009probabilistic}, when propagating weights to \aclp{SM}: \emph{How to distribute the weight of a fact to all the entailed \acp{SM}?}
\topicnote{parameters}
We represent non-unique choices by parameters that can be later estimated from further information, \ie\ data.
This approach enables later refinement from additional evidence and also scoring a program  w.r.t.\ the available data.

\topicnote{\acp*{ASP}}
\Acfp{ASP} are logic programs based on the \acl{SM} semantics of \acp{NP} \cite{lifschitz2002answer,lifschitz2008twelve}. 
\Acp{ASP} represent a problem and the resulting models (\emph{answer sets}) can be found using different approaches, including SAT solving technology \cite{gebser2011potassco,adrian2018asp,niemela1997smodels} or through top-down search \cite{alberti2017cplint,arias2020justifications,marple2017computing}.

\topicnote{\acs*{DS}}
The \ac{DS} \cite{sato1995statistical,riguzzi2022foundations} is the base for \acf{PLP}, to extend \aclp{LP} with probabilistic reasoning.
\sidenote{\franc{addressing} better examples that show the real usefulness of this approach \franc{ \eg\ a toy problem with a biased coin; also see the examples on the other systems.}}%
We are particularly interested in the following setting and application scenarios of such an extension to logic programs:
\begin{description}
	\item[Partial Observability] \topicnote{partial observability}A system's state can have \emph{hidden variables}, not reported by the sensors. 
	\item[Sensor Error] \topicnote{sensor error}Information gathered from the sensors can carry \emph{stochastic perturbations}.
	\item[Representation Induction] \topicnote{representation induction}Combine representations with data to induce \emph{more accurate representations}.
	\item[Probabilistic Tasks] \topicnote{probabilistic tasks}Support common probabilistic tasks such as \acf{MAP}, \ac{MLE} and \ac{BI} on the \emph{representation domain} \ie\ the set of all events.
\end{description}

\noindent The remainder of this article is structured as follows: \cref{sec:framework} provides necessary context.
In \cref{sec:syntax.and.semantics} we discuss the syntax and semantics of our proposed language for \acfp{WASP}.
We also define a weight function over total choices and address the issue of how to propagate these probabilities from facts to events, in \cref{sec:propagating.weights}.
This method relies on an equivalence relation on the set of events.
Furthermore, we express uncertainty by polynomial expressions over variables which depend on the total choices and on the stable models. By then the \emph{Partial Observability} and \emph{Sensor Error} points are addressed.  An evidence that our approach is sound is given by \cref{cor:prob.one} where we show that replacing certain facts (with weight $1.0$) by deterministic facts doe not change the probability.
Some final remarks and ideas for future developments including \emph{Representation Induction} and \emph{Probabilistic Tasks} are presented in \cref{sec:discussion}.

\section{Framework}
\label{sec:framework}

\topicnote{\ac*{WASP}}%
We start by refining \cref{asp:system.representation.data,asp:sensor.representation.events} to set `representation' as an `\ac{ASP} with weights':
\begin{assumption}[Representation by \Aclp*{ASP} with Weigths]\label{asp:wasp.represent.systems}
	
	\phantom{~}
	\begin{description}
		\item[\Aclp*{ASP} and Weights] A \emph{representation} of a system is an \emph{\acl{ASP} }that includes \emph{weighted facts}.
	\end{description}
	\topicnote{\acl*{WF}}
	A \emph{\acf{WF}} or an \emph{annotated fact} has the from `$\weightfact{a}{w}$' where $a$ is an atom, $w \in \intcc{0, 1}$, and defines the disjunctive fact $a \vee \co{a}$. A model will include either $a$ or $\co{a}$ but never both.	
\end{assumption}

\topicnote{\acl*{TC}}
Selecting one of $a, \co{a}$ for each \ac{WF} in a program will lead to a \emph{\acf{TC}}
\footnote{We use the term `choice' for historical reasons, \eg\ see \cite{cozman2020joy} even though not related to the usual `choice' elements, atoms or rules from \eg\ \cite{calimeri2020aspcore}.}.
\topicnote{propagating}
Propagating weights from \acp{WF} to \acp{TC} is relatively straightforward (see \cref{eq:weight.total.choice}) but propagation  to events requires a step trough the program's \aclp*{SM},  addressed in \cref{sec:propagating.weights}. 

\paragraph{About Propagating Weights from \Aclp{TC}.}%

Our goal to propagate weights from \acp{TC} to \acp{SM} and from there to any event soon faces a non-deterministic problem, illustrated by the program $\SBF$ in \cref{ex:fruitful}, \cref{ssec:propagating.weights}, where multiple \acp{SM}, $ab$ and $ac$, result from a single \ac{TC}, $a$, but \emph{there is not enough information in the representation to assign a single weight to each one of those \acp{SM}}.
In \cref{subsec:from.tchoices.to.events} we use algebraic variables\footnote{We explicitly write `algebraic' variables to avoid confusion with logic variables.} to describe the lack of information in a representation in order to deterministically propagate the weight to the \aclp{SM} and events.
The values of those variables can be estimated from available data, a contribute to the \emph{Representation Induction} application goal, set in \cref{sec:introduction}.

The lack of a unique \acl{SM} from a \acl{TC} is also addressed in~\cite{cozman2020joy} along an approach using credal sets.
In another related work~\cite{verreet2022inference}, epistemic uncertainty (or model uncertainty) is considered as lack of knowledge about the underlying model, that may be mitigated via further observations.
This seems to presuppose a bayesian approach to imperfect knowledge in the sense that having further observations allows one to improve or correct the model.
Indeed, that approach uses Beta distributions on the total choices in order to be able to learn a distribution on the events.
This approach seems to be specially fitted to being able to tell when some weight lies beneath some given value.
Our approach is similar in spirit, while remaining algebraic in the way the propagation of weights is addressed.

\section{Syntax and Semantics of Weighted ASP}
\label{sec:syntax.and.semantics}

We start the formal definition of \emph{\acl*{WASP}} with the setup and discussion of a minimal syntax and semantics of propositional \ac{ASP}, without variables, functors or relation symbols, but enough to illustrate our method to propagate weights from annotated facts to events. From now on `$\neg x$' and `$\co{x}$' denote classical negation and `$\naf x$' default negation.

\paragraph{Syntax.}
\label{par:syntax}

We slightly adapt \cite{calimeri2020aspcore}. Let $\ATOMSset$ be a finite set of symbols, the \emph{positive atoms}. 
For $a \in \ATOMSset$, the expressions $a$ and $\neg a$ (the later a \emph{negative atom}, also denoted $\co{a}$) are \emph{(classical) atoms}.
\defnote{atom}
If $a$ is an atom, the expressions $a$ and $\naf a$ are \emph{(naf-)literals}.
\defnote{literal}
A \textit{rule} is of the form
\defnote{rule}
$$
h_1 \disj \cdots \disj h_n \clause 
	b_1 \conj \cdots \conj b_m
$$
where the $h_i$ are atoms and the $b_j$ are literals. 
The symbol `$\clause$' separates the \textit{head} from the \textit{body}. 
\defnote{head;body}
A rule is a \emph{constraint}\footnote{An `\emph{integrity constraint}' in \cite{calimeri2020aspcore}.} if $n = 0$, \emph{normal} if $n = 1$, \emph{disjunctive} if $n > 1$, and a \emph{fact} if $m = 0$.
\defnote{constraint;normal;disjunctive;fact}

An \textit{\acf{ASP}} is a set $P$ of facts and rules, denoted, resp. $\FACTSset\at{P}$ and $\RULESset\at{P}$, or simply $\FACTSset$ and $\RULESset$.
\defnote{program}
In a \textit{normal program} all the rules are normal.
Notice that programs with constraint or disjunctive rules can be converted into normal programs \cite{gebser2022answer}.

\paragraph{Semantics.}

The standard semantics of an \ac{ASP} has a few different, but equivalent, definitions \cite{lifschitz2008twelve}.
A common definition is as follows \cite{gelfond1988stable}: let $P$ be a \acl{NP}.
The Gelfond/Lifschitz \emph{reduct} of $P$ relative to the set $X$ of atoms results from (i) deleting rules that contain a literal of the form $\naf p$ in the body with $p \in X$ and then (ii) deleting the remaining literals of the form $\naf q$ from the bodies of the remaining rules.
\defnote{reduct}
Now, $M$ is a \textit{\acf{SM}} of $P$ if it is the minimal model of the reduct of $P$ relative to $M$.
\defnote{\acl{SM}}
We denote by $\MODELset\at{P}$, or simply $\MODELset$, the set of \aclp{SM} of the program~$P$.

\paragraph{Evaluation without Grounding.}

While the most common form to generate \aclp{SM} is based on \emph{grounding}, a different approach is the one supported by \texttt{s(CASP)}, a system that can evaluate \ac{ASP} programs with function symbols (functors) and constraints without grounding them either before or during execution, using a method similar to SLD resolution~\cite{marple2017computing,arias2020justifications}.
\sidenote{Improve grounding and propositional cases.}
This enables the generation of human readable explanations of the results of programs and addresses two major issues of grounding-based solvers, that 
(i) either do not support function symbols or, using finite domains, lead to exponential groundings of a program and 
(ii) compute the complete model of the grounded program when, in some scenarios, it is desirable to compute only a partial stable model containing a query.

\subsection*{\acsp{WASP} and their Derived Programs}

\emph{\Acfp{WASP}} extend \acp{ASP} by adding facts with weight annotations, \ie\ \emph{\acfp{WF}}.
\defnote{\acl{WF}}
Notice that we have $w\in \intcc{0,1}$ but $w$ is interpreted as a \emph{balance} between the \emph{choices} $a$ and $\co{a}$, and \emph{not a probability}.

We denote the set of \aclp{WF} of a program $P$ by $\WEIGHTFset\at{P}$, and $\WATOMset\at{P}$ the set of positive atoms in $\WEIGHTFset$. When possible we simplify notation as $\WEIGHTFset, \WATOMset$.

Our definition of \acp{WASP} is restricted because our goal is to illustrate the core of a method to propagate weights from \aclp{TC} to events.
Our programs do not feature logical variables, relation symbols, functors or other elements common in standard \ac{ASP}.
Also, weight annotations are not associated to (general) rule heads or disjunctions.
However, these last two restrictions do not reduce the expressive capacity of the language because, for the former, a rule with an annotated head can be rewritten as:
\begin{equation*}
	\weightrule{\alpha}{w}{\beta} \qquad \Longrightarrow \qquad
	\left\{
		\begin{aligned}
		\weightfact{\gamma & }{w},                       %
		\\
		\alpha           & \clause \beta \wedge \gamma
	\end{aligned}
	\right.
\end{equation*}
while annotated disjunctive facts
\begin{equation*}
	\weightfact{\alpha \vee \beta}{w} \qquad \Longrightarrow \qquad
	\left\{
	\begin{aligned}
		\weightfact{\gamma  & }{w},           %
		\\
		\alpha \vee \beta & \clause \gamma,           %
		\\
		\co{\alpha}& \clause \co{\gamma},  &&
		\co{\beta}& \clause \co{\gamma}.
	\end{aligned}
	\right.
\end{equation*}

\paragraph{Derived Program.}

The \emph{derived program} of a \ac{WASP} is the \ac{ASP} obtained by replacing each \acl{WF} $\weightfact{a}{w}$ by a disjunction $a \disj \co{a}$.
\defnote{derived}
The \textit{\aclp{SM}} of a \acs{WASP} program are the \aclp{SM} of its derived program.
So, we also denote the set of \acp{SM} of a (derived or) \acs{WASP} program $P$ by $\MODELset\at{P}$ or $\MODELset$.

\paragraph{Events.}

An \emph{event} of a program $P$ is a set of atoms from $P$.
\defnote{event}
We denote the set of events by
$\EVENTSset\at{P}$ or simply $\EVENTSset$.
An event $e \in \EVENTSset$ which includes a set $\set{x, \co{x}}$ is said to be \textit{inconsistent}; otherwise it is \textit{consistent}.
\defnote{(in)consistent}
The set of consistent events is denoted by $\CONSISTset$.

\begin{example}
	\label{ex:fruitful}\em

	Consider the following \acl{WASP} :
	\begin{equation}\label{eq:fruitful}
		\SBF = \left\{\begin{split}
			\weightfact{a & }{0.3},   %
			\\
			b \vee c    & \clause a
		\end{split}
		\right.
	\end{equation}
	which has the set $\WEIGHTFset = \{ a:0.3 \}$ of \aclp{WF}.
This program is transformed into the \acl{ASP}
	\begin{equation}\label{eq:derived.fruitful}
		\SBF' = \left\{\begin{split}
			a \vee \co{a} & ,          %
			\\
			b \vee c      & \clause a,
		\end{split}
		\right.
	\end{equation}
	with the set
	$ \MODELset = \{\, \co{a}, ab, ac\, \} $
	of stable models.

The atoms of these programs are 
\begin{equation}
	\ATOMSset = \set{a, \co{a}, b, \co{b}, c, \co{c}}
	\label{eq:atoms.fruitful}
\end{equation}
and the events are%
\footnote{$\powerset{X}$ is the \emph{power set} of $X$: $A \in \powerset{X} \Leftrightarrow A \subseteq X$.} 
\begin{equation}
	\EVENTSset = \powerset{\ATOMSset}.
	\label{eq:fruitful.events}
\end{equation}
\end{example}

\subsection*{Total Choices and their Weights}
\label{ssec:totalchoices.weights}

A disjunctive head $a \disj \co{a}$ in the derived program represents a single \textit{choice}, either $a$ or $\co{a}$.
\defnote{\acl{TC}}
We define the set of \emph{\acfp{TC}} of a set of atoms by the recursion 
\begin{equation}
\left\{
\begin{aligned}
	\TCHOICEset\at{\emptyset} &= \emptyset, \\
	\TCHOICEset\at{X \cup a} &= 
		\bigcup_{t \in \TCHOICEset\at{X}} \del{t \cup a}
		\quad \cup \quad
		\bigcup_{t \in \TCHOICEset\at{X}} \del{t \cup \co{a}}
\end{aligned}
\right.\label{eq:def.total.choice}	
\end{equation}
where $X$ is a set of atoms and $a$ is an atom. The \aclp{TC} of a \ac{WASP} $P$ are the \acp{TC} of it's positive atoms: $\TCHOICEset\at{P} = \TCHOICEset\at{\WATOMset\at{P}}$. When possible we write simply $\TCHOICEset$. 






Given a \ac{WASP}, the \emph{weight of the \acl{TC} $t \in \TCHOICEset$} is given by the product
\defnote{$\wgtT$}
\begin{equation}
	\wgtT\at{t} =
	\prod_{\substack{
			\weightfact{a}{w}~ \in ~\WEIGHTFset,%
	\\
			a~ \in~ t}} w\quad \times\quad
	\prod_{\substack{
			\weightfact{a}{w}~ \in ~\WEIGHTFset,%
	\\
			\co{a}~\in~ t}} \co{w}.
	\label{eq:weight.total.choice}
\end{equation}

Here $\co{w} = 1 - w$, and we use the subscript in $\wgtT$ to explicitly state that this function concerns total choices.
Later we'll use subscripts $\MODELset, \EVENTSset$ to deal with weight functions of \aclp{SM} and events, $\wgtM, \wgtE$.

Some \aclp{SM} are entailed from some \aclp{TC} while other \acp{SM} are entailed by other \acp{TC}.
We write $\tcgen{t}$ to represent the set of \aclp{SM} entailed by the \acl{TC} $t \in \TCHOICEset$.

Our goal can now be rephrased as to know how to propagate the weights of the program's \aclp{TC}, $\wgtT$, in \cref{eq:weight.total.choice}
to the program's events, $\wgtE$ to be defined later, in \cref{eq:weight.events,eq:weight.events.unconditional}.

\paragraph{Propagation of Weights.}

As a first step to propagate weight from \aclp{TC} to events, consider the $\SBF$ program of \cref{eq:fruitful} and a possible propagation of $\wgtT:\TCHOICEset \to \intcc{0,1}$ from \aclp{TC} to the \aclp{SM}, $\wgtM:\MODELset \to \intcc{0,1}$ (still informal, see \cref{eq:weight.stablemodel}).
%
%
It might seem straightforward, in \cref{ex:fruitful}, to set $\wgtM\at{\co{a}}=0.7$ but there is no explicit way to assign values to $\wgtM\at{ab}$ and $\wgtM\at{ac}$.
We represent this non-determinist by a parameter $\theta$ as in
\begin{equation}
	\begin{aligned}
		\wgtM\at{ab} & = 0.3\, \theta,
		\\
		\wgtM\at{ac} & = 0.3\, (1 - \theta)
	\end{aligned}\label{eq:theta.and.stablemodels}
\end{equation}
to express our knowledge that $ab$ and $ac$ are models entailed from a specific choice and, simultaneously, the inherent non-determinism of that entailment.
In general, it might be necessary to have several such parameters, each associated to a given \acl{SM} $s$ (in \cref{eq:theta.and.stablemodels}, $s = ab$ in the first line and $s = ac$ in the second line) and a \acl{TC} $t$ ($t=a$ above), so we write $\theta_{s,t}$.
Obviously, for reasonable $\theta_{s,t}$, the total choice $t$ must be a subset of the stable model $s$.

Unless we introduce some bias, such as $\theta = 0.5$ as in \lpmln\ \cite{lee2016weighted}, the value for $\theta_{s,t}$ can't be determined just with the information given in the program. But it might be estimated with the help of further information, such as an empirical distribution from a dataset.
Further discussion of this point is outside the scope of this paper.

Now consider the program
\begin{equation}\label{eq:single}
	\left\{\begin{split}
		\weightfact{a & }{0.3},   %
		\\
		b & \clause a \wedge \naf b
	\end{split}
	\right.
\end{equation}
that has a single \ac{SM}, $\co{a}$. Since the weights are not interpreted as probabilities, there is no need to have the sum on the \aclp{SM} equal to $1$. So the weights in the \acp{TC} of \cref{eq:single} only set
$$
\wgtM\at{\co{a}} = 0.7.
$$
In this case, if we were to derive a probability of the \acp{SM}, normalization would give $\pr{\co{a}} = 1.0$.

Also facts without annotations can be transformed into facts with weight $1$:
\begin{equation}
	a \qquad \Longrightarrow \qquad \weightfact{a}{1.0}. \label{eq:noannotation}
\end{equation}

\topicnote{propagation, semantics}
The method that we are proposing does not follow the framework of~\cite{kifer1992theory} and others, where the syntax of the program determines the propagation from probabilities explicitly set either in facts or other elements of the program.
Our approach requires that we consider the semantics, \emph{i.e.}\ the \aclp{SM} of the program, independently of the syntax that provided them. From there we propagate weights to the program's events and then, if required, normalization provides the final probabilities.
Moreover, we allow the occurrence of variables in the weights, in order to deal with the non-determinism that results from the non-uniqueness of \acp{SM} entailed from a single \ac{TC}.
These variables can be later estimated from available data.%

\subsection*{Related Approaches and Systems}
\label{ssec:other.approaches}

The core problem of setting a semantics for probabilistic logic programs, the propagation of probabilities from \aclp{TC} to \aclp{SM} in the case of \ac{ASP} or to other types in other logic programming systems (\eg\ to possible worlds in \texttt{Problog}) has been studied for some time~\cite{kifer1992theory,sato1995statistical}.

For example, the \emph{credal set} approach of~\cite{cozman2020joy}, defines $\prT$ in a way similar to \cref{eq:weight.total.choice} but then, for $a \in \ATOMSset, t \in \TCHOICEset$, the probability $\pr{a \given t}$ is unknown but bounded by $\underbar{\prfunc}\at{a \given t}$ and $\overline{\prfunc}\at{a \given t}$, that can be explicitly estimated from the program.

\texttt{Problog} \cite{fierens2015inference,verreet2022inference} extends \texttt{Prolog} with probabilistic facts so that a program specifies a probability distribution over possible worlds.
A \textit{world} is a model of $T \cup R$ where $T$ is a total choice and $R$ the set of rules of a program.
The semantics is only defined for \textit{sound} programs~\cite{riguzzi2013well} \ie, programs for which each possible total choice $T$ leads to a well-founded model that is two-valued or \textit{total}.
The probability of a possible world that is a model of the program is the probability of the total choice.
Otherwise the probability is $0$~\cite{riguzzi2013well,van1991well}.

Another system, based on Markov Logic~\cite{richardson2006markov}, is \lpmln~\cite{lee2016weighted,lee2017lpmln}, whose models result from \textit{weighted rules} of the form $a \clause b \wedge n$ where $a$ is disjunction of atoms, $b$ is conjunction of atoms and $n$ is constructed from atoms using conjunction, disjunction and negation.
For each model there is a unique maximal set of rules that are satisfied by it and the respective weights determine the weight of that model, that can be normalized to a probability.

\subsection*{Towards Propagating Weights from Total Choices to Events}
\label{ssec:propagating.weights}

The program $\SBF$ in \cref{eq:fruitful} from \cref{ex:fruitful} showcases the problem of propagating weights from \aclp{TC} to \aclp{SM} and then to events.
The main issue arises from the lack of information in the program on how to assign \emph{un-biased} weights to the \aclp{SM}.
This becomes crucial in situations where multiple \aclp{SM} result from a single \acl{TC}.

Our \cref{asp:system.representation.data,asp:sensor.representation.events,asp:wasp.represent.systems} enunciate that a \ac{WASP} program represents a \emph{system}; the \emph{states} of that system, which are partially observable and stochastic, are associated to the program's \aclp{SM}; and state \emph{observations} are encoded as events, \ie\ sets of atoms of the program. 
Then:
\begin{enumerate}

	\item With a weight set for the \aclp{SM}, we extend it to any event in the program domain.

	\item In the case where some statistical knowledge is available, for example, in the form of a distribution relating some atoms, we consider it as `external' knowledge about the parameters, that doesn't affect the propagation procedure described below.

	\item However, that knowledge can be used to estimate the parameters $\theta_{s,t}$ and to `score' the program.

	\item\label{item:program.selection} In that case, if a program is but one of many possible candidates, then that score can be used, \eg\ as fitness, by algorithms searching (optimal) programs of a dataset of events.

	\item If events are not consistent with the program, then we ought to conclude that the program is wrong and must be changed accordingly.

\end{enumerate}

Next, we address the problem of propagating a weight, possibly using parameters (\ie\ algebraic variables) such as $\theta$ in \cref{ex:fruitful}, defined on the \aclp{SM} of a program, $\pwM: \MODELset \to \mathbb{R}$, to the events of that program: $\pwE: \EVENTSset \to \mathbb{R}$.
The latter function can then be normalized and set a probability $\prE:\powerset{\EVENTSset} \to \intcc{0,1}$.
This way probabilistic reasoning is consistent with the \ac{ASP} program and follows our interpretation of \aclp{SM} as the states of an observable system.

\section{Propagating Weights}
\label{sec:propagating.weights}
\begin{figure}[t]
	\begin{center}
		\begin{tikzpicture}[node distance=2em]
			\node[event] (E) {$\emptyevent$};
			\node[event, above = of E] (c) {$c$};
			\node[tchoice, left = of c] (a) {$a$};
			\node[event, left = of a] (b) {$b$};
			\node[right = of c] (invis) {};
			\node[smodel, above = of b] (ab) {$ab$};
			\node[smodel, above = of c] (ac) {$ac$};
			\node[event, above right = of ab] (abc) {$abc$};
			\node[event, above left = of ab] (abC) {$\co{c}ab$};
			\node[event, above right = of ac] (aBc) {$\co{b}ac$};
			\node[indep, right = of ac] (bc) {$bc$};
			\node[tchoice, smodel, right = of invis] (A) {$\co{a}$};
			\node[event, right = of bc] (Ac) {$\co{a}c$};
			\node[event, right = of aBc] (Abc) {$\co{a}bc$};
			\draw[doubt] (a) to[bend left] (ab);
			\draw[doubt] (a) to[bend right] (ac);

			\draw[doubt] (ab) to[bend left] (abc);
			\draw[doubt] (ab) to[bend right] (abC);

			\draw[doubt] (ac) to[bend right] (abc);
			\draw[doubt] (ac) to[bend left] (aBc);

			\draw[doubt, dashed] (Ac) to (Abc);

			\draw[doubt] (A) to (Ac);
			\draw[doubt] (A) to[bend right=60] (Abc);

			\draw[doubt] (ab) to[bend right] (E);
			\draw[doubt] (ac) to[bend left=45] (E);
			\draw[doubt] (A) to[bend left] (E);

			\draw[doubt] (ab) to (b);
			\draw[doubt] (ac) to (c);
			\draw[doubt, dashed] (c) to[bend right] (bc);
			\draw[doubt, dashed] (abc) to[bend left] (bc);
			\draw[doubt, dashed] (bc) to (Abc);
			\draw[doubt, dashed] (c) to[bend right] (Ac);
		\end{tikzpicture}
	\end{center}

	\caption{%
		This (partial sub-/super-set) diagram shows some events related to the \aclp{SM} of the program $\SBF$.
		The circle nodes are \aclp{TC} and shaded nodes are \aclp{SM}.
		Solid lines represent relations with the \aclp{SM} and dashed lines some sub-/super-set relations with other events.
		The set of events contained in all \aclp{SM}, denoted by $\consequenceclass$, is $\set{ \emptyevent }$ in this example, because $\co{a} \cap ab \cap ac = \emptyset = \emptyevent$.}
	\label{fig:ex:fruitful}
\end{figure}

The diagram in \cref{fig:ex:fruitful} illustrates the problem of propagating weights from \aclp{TC} to \aclp{SM} and then to general events in an \emph{edge-wise} process, \ie\ where the value in a node is defined from the values in its neighbors.
This quickly leads to interpretation issues concerning weight, with no clear systematic approach.
For example, notice that $bc$ is not directly related with any \acl{SM}.
Propagating values through edges would assign a value ($\not= 0$) to $bc$ hard to explain in terms of the semantics of the program.
Instead, we propose to settle such propagation on the relation an event has with the \aclp{SM}.

\subsection{An Equivalence Relation}
\label{subsec:equivalence.relation}
\begin{figure}[t]
	\begin{center}
		\begin{tikzpicture}[node distance=2em]
			\node[event] (E) {$\emptyevent$};
			\node[event, above = of E] (c) {$c$};
			\node[tchoice, left = of c] (a) {$a$};
			\node[event, left = of a] (b) {$b$};
			\node[right = of c] (invis) {};
			\node[smodel, above = of b] (ab) {$ab$};
			\node[smodel, above = of c] (ac) {$ac$};
			\node[event, above right = of ab] (abc) {$abc$};
			\node[event, above left = of ab] (abC) {$\co{c}ab$};
			\node[event, above right = of ac] (aBc) {$\co{b}ac$};
			\node[indep, right = of ac] (bc) {$bc$};
			\node[tchoice, smodel, right = of invis] (A) {$\co{a}$};
			\node[event, right = of bc] (Ac) {$\co{a}c$};
			\node[event, right = of aBc] (Abc) {$\co{a}bc$};
			\path[draw, rounded corners, pattern=north west lines, opacity=0.2]
			(ab.west) -- (ab.north west) --
			(abC.south west) -- (abC.north west) -- (abC.north) --
			(abc.north east) -- (abc.east) -- (abc.south east) --
			(ab.north east) -- (ab.east) -- (ab.south east) --
			(a.north east) --
			(E.north east) -- (E.east) --
			(E.south east) -- (E.south) --
			(E.south west) --
			(b.south west) --
			(ab.west) ;
			\path[draw, rounded corners, pattern=north east lines, opacity=0.2]
			(ac.south west) -- (ac.west) -- (ac.north west) --
			(abc.south west) -- (abc.west) -- (abc.north west) --
			(aBc.north east) -- (aBc.east) -- (aBc.south east) --
			(ac.north east) --
			(c.east) --
			(E.east) -- (E.south east) -- (E.south) -- (E.south west) --
			(a.south west) -- (a.west) -- (a.north west) -- (a.north) --
			(ac.south west) ;
			\path[draw, rounded corners, pattern=vertical lines, opacity=0.2]
			(Ac.west) --
			(Abc.north west) -- (Abc.north) --
			(Abc.north east) -- (Abc.south east) --
			(Ac.north east) -- (Ac.south east) --
			(A.east) -- (A.south east) --
			(E.south east) -- (E.south) --
			(E.south west) -- (E.west) --
			(E.north west) --
			(Ac.south west) -- (Ac.west);
		\end{tikzpicture}
	\end{center}

	\caption{%
		Classes of (consistent) events related to the \aclp{SM} of $\SBF$ are defined through sub-/super-set relations. 
		In this picture we can see, for example, that $\set{\co{c}ab, ab, b}$ and $\set{a, abc}$ are part of different classes, represented by different fillings.
		As before, the circle nodes are \aclp{TC} and shaded nodes are \aclp{SM}.
		Notice that $bc$ is not in a filled area.}
	\label{fig:ex:fruitful.classes}
\end{figure}
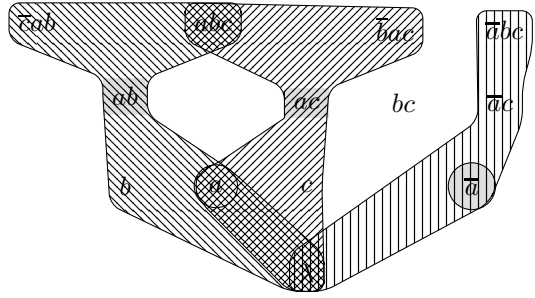

\begin{figure}[t]
	\begin{center}
		\begin{tikzpicture}[3d view]
			\node[event] (INDEPENDENT) at (0,0,0){$\indepclass$};
			\node[smodel] (A) at (0,0,2) {$\co{a}$};
			\node[smodel] (ab) at (3,0,0) {$ab$};
			\node[smodel] (ac) at (0,3,0) {$ac$};
			\node[event] (Aab) at (3,0,2) {$\co{a},ab$};
			\node[event] (Aac) at (0,3,2) {$\co{a},ac$};
			\node[event] (abac) at (3,3,0) {$ab,ac$};
			\node[event] (Aabac) at (3,3,2) {$\consequenceclass$};
			\node[event] (INCONSISTENT) at (-4, 0, 2) {$\inconsistent$
				(inconsistent)};
			\draw[->] (INDEPENDENT) -- (A);
			\draw[->] (INDEPENDENT) -- (ab);
			\draw[->] (INDEPENDENT) -- (ac);
			\draw[->] (A) -- (Aab);
			\draw[->] (A) -- (Aac);
			\draw[->] (ab) -- (Aab);
			\draw[->] (ab) -- (abac);
			\draw[->] (ac) -- (Aac);
			\draw[->] (ac) -- (abac);
			\draw[->] (Aab) -- (Aabac);
			\draw[->] (Aac) -- (Aabac);
			\draw[->] (abac) -- (Aabac);
		\end{tikzpicture}
	\end{center}

	\caption{%
		Lattice of the \aclp{SC} from \cref{ex:fruitful}.
		In this diagram the nodes are the different \aclp{SC} that result from the \aclp{SM}, plus the \emph{inconsistent} class ($\inconsistent$).
		The bottom node ($\indepclass$) is the class of \emph{independent} events, those that have no sub-/super-set relation with any \ac{SM} and the top node ($\consequenceclass$) represents events related with all the \acp{SM} \ie\ the \emph{consequences} of the program.
		As in previous diagrams, shaded nodes represent the \acp{SM}.}
	\label{fig:fruitful.lattice}
\end{figure}

Our path to propagate weights starts with the perspective that \aclp{SM} play a role similar to \emph{prime factors} or \emph{principal ideals}.
The \aclp{SM} of a program play a role akin to `irreducible events' entailed from that program and any event must be considered under its relation with the \aclp{SM}.

From \cref{ex:fruitful} (\ie\ $\SBF$) and \cref{fig:ex:fruitful.classes} consider the \aclp{SM} $\co{a}, ab, ac$ and events $a, abc$ and $c$.
While $a$ is related with (i.e.~contained in) both $ab, ac$, the event $c$ is related only with $ac$.
So, $a$ and $c$ are \emph{related with different sets of \aclp{SM}.}
On the other hand, $abc$ contains both $ab, ac$.
Therefore $a$ and $abc$ are \emph{related with the same set of \aclp{SM}.}
We proceed to formalize this relation.

The \textit{\acf{SC}} of the event $e\in \EVENTSset$ is\defnote{$\stablecore{e}$}
\begin{equation}
	\stablecore{e} := \set{\, s \in \MODELset \given s \subseteq e \vee e \subseteq s\, } \label{eq:stable.core}
\end{equation}
where $\MODELset$ is the set of \aclp{SM}.

Notice that the minimality of \aclp{SM} implies that either $e$ is a \acl{SM} or at least one of $\exists s \del{s \subseteq e}, \exists s \del{e \subseteq s}$ is false \emph{i.e.,}\ no \acl{SM} contains another.

We now define an equivalence relation so that two events are related if either both are inconsistent or both are consistent and, in the latter case, with the same \acl{SC}.
\begin{definition}[Equivalence Relation on Events]\label{def:equiv.rel}

	For a given program, let $u, v \in \EVENTSset$.
	The equivalence relation\defnote{$\sim$}
	$u \sim v$ is defined by
	\begin{equation}
		u,v \not\in\CONSISTset \vee \del{u,v \in \CONSISTset \wedge \stablecore{u} = \stablecore{v}}.\label{eq:equiv.rel}
	\end{equation}

\end{definition}

This equivalence relation defines a partition on the set of events, where each class holds a unique relation with the \aclp{SM}.
In particular we denote each class by:\defnote{$\class{e}$}
\begin{equation}
	\class{e} =
	\begin{cases}
		\inconsistent := \EVENTSset \setminus \CONSISTset
		 & \text{if~} e \in \EVENTSset \setminus \CONSISTset, %
		\\
		\set{u \in \CONSISTset \given \stablecore{u} = \stablecore{e}}
		 & \text{if~} e \in \CONSISTset.
	\end{cases}\label{eq:event.class}
\end{equation}
where $\inconsistent$ denotes the set $\EVENTSset\setminus\CONSISTset$ of \emph{inconsistent} events, \ie\ events that contain $\set{x,\co{x}}$ for some atom $x$.\defnote{$\inconsistent$}

\begin{proposition}[Class of the Program's Consequences]
	\label{prop:consequence.class}
	\defnote{$\emptyevent,\consequenceclass$}
	Let $\emptyevent$ be the empty set event (notice that $\emptyevent = \emptyset \in \EVENTSset$)
	\footnote{We adopt the notation `$\emptyevent$' for \emph{empty word}, from formal languages, to distinguish `$\emptyset \in \EVENTSset$' from `$\emptyset \subset\EVENTSset$'.}, and $\consequenceclass$ the \emph{consequence class} of (consistent) events related with all the
	\aclp{SM}.
	Then
	\begin{equation}
		\class{\emptyevent} = \consequenceclass.
	\end{equation}
\end{proposition}
\begin{proof}
	Let $x \in \class{\emptyevent}$ be consistent.
	Then $x \sim \lambda$. Since $\stablecore{\emptyevent} = \MODELset$ also $\stablecore{x} = \MODELset$. Hence  $x\in\consequenceclass$.

	Now, let $x \in \consequenceclass$. So $x$ is consistent and, for each $s \in \MODELset$, either $s \subseteq x$ or $s \supseteq x$. So, $\stablecore{x} = \MODELset = \stablecore{\emptyevent}$. By \cref{def:equiv.rel}, $x \sim \emptyevent$ \ie\ $x \in \class{\emptyevent}$.
	\hfill
\end{proof}

The combinations of \aclp{SM}, \ie~the \aclp{SC}, together with the set of inconsistent events ($\inconsistent$) forms a set of representatives for the equivalence relation $\sim$.
Since all events within a consistent equivalence class have the same \acl{SC}, we are interested in functions (including weight assignments), that are constant within classes.
A function $f:\EVENTSset\to Y$, where $Y$ is any set, is said to be \emph{coherent} if
\begin{equation}
	\forall e\in \EVENTSset~\forall u\in \class{e} \left( f\at{u} = f\at{e} \right).
\end{equation}

Considering coherent functions, in the specific case of \cref{eq:fruitful}, instead of dealing with the $2^6 = 64$ events, we need to consider only the
$2^3 + 1 = 9$ classes, well defined in terms of combinations of the \aclp{SM}, to define coherent functions.
In general, a program with $n$ atoms and $m$ \aclp{SM} has $2^{2n}$ events and $2^m + 1$ \aclp{SC} --- but easily $m \gg n$.

\subsection{From Total Choices to Events}
\label{subsec:from.tchoices.to.events}

The `propagation' phase, traced by \cref{eq:weight.total.choice} and \crefrange{eq:weight.stablemodel}{eq:weight.events.unconditional}, starts with the weight of \aclp{TC}, $\wgtT\at{t}$, propagates it to the \aclp{SM}, $\wgtm{s}$, and then, within the equivalence relation from \cref{eq:equiv.rel}, to a coherent weight of events,
$\wgte{e}$.
So we are specifying a sequence of functions
\begin{equation}
	\wgtT \longrightarrow \wgtM \longrightarrow \wgtC \longrightarrow \wgtE\label{eq:sequence.functions}
\end{equation}
on successive larger domains
$$
\TCHOICEset \longrightarrow \MODELset \longrightarrow \class{\EVENTSset} \longrightarrow \EVENTSset
$$
such that the last function ($\wgtE$) is a finite coherent function on the set of events and thus, as a final step, it can easily be used to define a probability distribution of events by normalization:
\(
\wgtE \longrightarrow \prE
\).

\subsubsection*{\Aclp{TC} and \Aclp{SM}}
\label{par:prop.totalchoices}

Let's start by looking into the first two steps of the sequence of functions \cref{eq:sequence.functions}: $\wgtT$ and $\wgtM$.

The weight $\wgtT$ of the \acl{TC} $t \in \TCHOICEset$ is already given by \cref{eq:weight.total.choice}.

Recall that each \acl{TC} $t \in \TCHOICEset$, together with the rules and the other facts of a program, defines the set \tcgen{t} of \aclp{SM} associated with that choice.
Given a \acl{TC} $t \in \TCHOICEset$, a \acl{SM}
$s \in \MODELset $, and formal variables or values $\theta_{s,t} \in \intcc{0, 1}$ such that $\sum_{s\in \tcgen{t}} \theta_{s,t} = 1$, we define\defnote{$\wgtm{s,t}$}
\begin{equation}
	\wgtm{s, t} := \begin{cases}
		              \theta_{s,t} & \text{if~} s \in \tcgen{t}\cr 0 & \text{otherwise.}
	              \end{cases}
	\label{eq:weight.stablemodel}
\end{equation}

The $\theta_{s,t}$ parameters in \cref{eq:weight.stablemodel} express the \emph{program's} lack of information about the weight assignment, when a single \acl{TC} entails more than one \acl{SM}.
We address this issue by assigning a possibly unknown parameter, \ie~a formal algebraic variable ($\theta_{s,t}$) associated with a \acl{TC} ($t$) and a \acl{SM} ($s$).
This allows the expression of a quantity that does not result from the program's syntax but can be determined or estimated given more information, \eg\ observed data.

As sets, two \aclp{SM} can have non-empty intersection.
But because different \acp{SM} represent different states of a system ---which are \emph{disjoint events}--- we assume that the algebra of the \aclp{SM} is $\sigma$-additive

\begin{assumption}[\Aclp{SM} as Disjoint events]
	\label{assumption:smodels.disjoint}%

	For any set $X$ of \aclp{SM} and any \acl{TC} $t$,
	\begin{equation}
		\wgtM\at{X, t} = \sum_{s\in X}\wgtM\at{s, t}.\label{eq:smodels.disjoint}
	\end{equation}

\end{assumption}

\Cref{eq:smodels.disjoint} is the basis for \cref{eq:weight.class.consistent} and effectively extends $\wgtM:\MODELset \times \TCHOICEset \to \mathbb{R}$ to $\wgtM:\powerset{\MODELset} \times \TCHOICEset \to \mathbb{R}$.
Now the pre-condition of \cref{eq:weight.stablemodel} can be stated as $\wgtM\at{\MODELset\at{t}, t} = 1$.

\subsubsection*{Classes}
\label{par:prop.class.cases}

Consider the next step in sequence \cref{eq:sequence.functions}, the function $\wgtC$
on $\class{\EVENTSset}$.
Each class of the equivalence relation $\sim$ (eq.~\ref{eq:equiv.rel}) is either the inconsistent class ($\inconsistent$) or is associated with a \acl{SC}, \textit{i.e.~}a set of \aclp{SM}.
Therefore, $\wgtC$ is defined considering the following two cases:\defnote{$\wgtC$}
\paragraph{Inconsistent class.} This class contains events that are (classically) inconsistent, thus should never be observed and thus have weight zero:
\begin{equation}
	\wgtc{\inconsistent, t} := 0.
	\footnote{This weight being zero is independent of the \aclp{SM}.}
	\label{eq:weight.class.inconsistent}
\end{equation}
\paragraph{Consistent classes.} For the propagation function to be coherent, it must be constant within a class and its value dependent only on the \acl{SC}:
\begin{subequations}
	\begin{equation}
		\wgtc{\class{e}, t} := \wgtm{\stablecore{e}, t} = \sum_{s\in\stablecore{e}}\wgtm{s, t}.
		\label{eq:weight.class.consistent}
	\end{equation}
	and we further define the following:
	\begin{equation}
		\wgtc{\class{e}} := \sum_{t \in \TCHOICEset} \wgtt{t}\wgtc{\class{e}, t}
		\label{eq:weight.class.unconditional}
	\end{equation}
\end{subequations}
\Cref{eq:weight.class.consistent} states that the weight of a class $\class{e}$ is the weight of its \acl{SC}
($\stablecore{e}$) and \cref{eq:weight.class.unconditional}
\emph{averages} \cref{eq:weight.class.consistent} over the \aclp{TC}.
Notice that \cref{eq:weight.class.consistent} also applies to the independent class, $\indepclass = \set{ e \given \stablecore{e} = \emptyset}$, because events in this class are not related with any \acl{SM}:\defnote{$\indepclass$}
\begin{equation}
	\wgtc{\indepclass, t} = \sum_{s\in\emptyset}\wgtm{s, t} = 0.
	\label{eq:weight.class.independent}
\end{equation}

\subsubsection*{Events and Probability}
\label{par:propagation.event.cases}

Each consistent event $e \in \EVENTSset$ is in the class defined by its
\acl{SC} $\stablecore{e}$.
So, denoting the number of elements in $X$ as $\#
	X$, we set:\defnote{$\wgtE$}
\begin{subequations}
	\begin{equation}
		\wgte{e, t} :=
		\begin{cases}
			\frac{\wgtc{\class{e}, t}}{\# \class{e}} & \text{if~}\# \class{e} > 0, %
			\\
			0                                       & \text{otherwise}.
		\end{cases}
		\label{eq:weight.events}
	\end{equation}
	and, by averaging over the \aclp{TC}:
	\begin{equation}
		\wgte{e} := \sum_{t\in\TCHOICEset} \wgtt{t}\wgte{e, t}.
		\label{eq:weight.events.unconditional}
	\end{equation}
\end{subequations}

The \cref{eq:weight.events.unconditional} is the main goal of this paper: propagate the weights associated to facts of an \ac{WASP} to the set of all events of that program.
In order to get a probability from \cref{eq:weight.events.unconditional}, concerning the \emph{Probabilistic Tasks} goal, we define the \emph{normalizing factor}:
\begin{equation}
	Z :=
	\sum_{e \in \EVENTSset} \wgte{e} =
	\sum_{\class{e} \in \class{\EVENTSset}} \wgtc{\class{e}},\label{eq:normalizing.factor}
\end{equation}
and now \cref{eq:weight.events.unconditional} provides a straightforward way to define the \emph{probability of a single event $e \in \EVENTSset$}:\defnote{$\prE$}
\begin{equation}
	\prE\at{e} := \frac{\wgte{e}}{Z}.\label{eq:probability.event}
\end{equation}

\Cref{eq:probability.event} defines a coherent \emph{prior}\footnote{In the Bayesian sense that future observations might update this probability.} probability of events and, together with external statistical knowledge, can be used to learn about the \emph{initial} probabilities of the atoms.

\paragraph{The effect of propagation.} One way to assess the effect of propagating weights through the \acp{SM} to the events is to compare $\prE$ with a `probability' induced syntactically from the weights of facts. It is sufficient to consider what appens with the \aclp{TC}: if we syntactically induce a probability for the \acp{TC}, say $\prT$, is it true that $\prE\at{t} = \prT\at{t}$ for all $t \in \TCHOICEset$?

The weight assignment of the \aclp{TC} can also be normalized into a probability distribution. 
For $t\in\TCHOICEset$,\defnote{$\prT$}
\begin{equation}
	\prT\at{t} = \frac{\wgtt{t}}{\sum_{\tau \in \TCHOICEset} \wgtt{\tau}}
	\label{eq:syntactic.probability.totalchoices}
\end{equation}

And now we ask if these probabilities coincide in $\TCHOICEset$:

$$
\forall t \in \TCHOICEset \del{\prE\at{t} = \prT\at{t}}?
$$

It is easy to see that, in general, this cannot be true.
While the domain of $\prT$ is the set of \aclp{TC}, for $\prE$ the domain is much larger, including all the events.
Except for trivial programs, some events other than \aclp{TC} will have non-zero weight.

\begin{proposition}[Two Distributions] \label{prop:two.distributions} %

	If a program has a consistent event $e \in \CONSISTset\setminus\TCHOICEset$ such that $\prE\at{e} \not= 0$ then there is at least one $t\in\TCHOICEset$ such that
	\begin{equation}
		\prT\at{t} \not= \prE\at{t}. \label{eq:two.distributions}
	\end{equation}

\end{proposition}
\begin{proof}
	Suppose towards a contradiction that $\prT\at{t} = \prE\at{t}$ for
	all $t \in \TCHOICEset$ and $\eta$ as above.

	Then
	\begin{equation*}
		1 = \sum_{t \in \TCHOICEset} \prT\at{t}  = \sum_{t \in \TCHOICEset} \prE\at{t}.
	\end{equation*}

	Therefore, $\prE\at{\eta} = 0$ for all $\eta \in \CONSISTset\setminus\TCHOICEset$, in contradiction with the hypothesis on $e$.
	\hfill
\end{proof}

The essential, perhaps \emph{counter-intuitive}, conclusion of \cref{prop:two.distributions} is that we are dealing with \emph{two distributions}: $\prT$, restricted to the \aclp{TC}, results \emph{syntactically} from the annotations of the program, while $\prE$, extended to events, results from both the annotations and the program's \emph{semantics}, \ie\ the \aclp{SM}.

For \cref{ex:fruitful}:
\begin{equation*}
	\begin{aligned}
		\prT\at{a} & = 0.3          &  &
		\text{from the program } \SBF, %
		\\
		\prE\at{a} & = \frac{3}{64} &  &
		\text{from \cref{eq:sbf.prior}}.
	\end{aligned}
\end{equation*}

Now $\prE:\EVENTSset \to \intcc{0,1}$ can be extended to
$\prE:\powerset{\EVENTSset}\to\intcc{0,1}$ by abusing notation and setting, for $X \subseteq \EVENTSset$,
\begin{equation}
	\prE\at{X} = \sum_{x\in X}\prE\at{x}.
	\label{eq:prob.event.set}
\end{equation}
It is straightforward to verify that the latter satisfies the Kolmogorov axioms of probability.

We can now properly state the following property about \emph{certain facts}
such as $\weightfact{a}{1.0}$.
\begin{theorem}[Weight of Certain Facts]
	\label{teo:prob.one}

	Consider a program $A$ with the \acl{WF} $\weightfact{\alpha}{1.0}$ and
	$B$ that results from $A$ by replacing that fact by the deterministic fact $\alpha$. Then
	\begin{equation}
		\forall e \in \EVENTSset \left(\wgtE^A\at{e} = \wgtE^B\at{e}\right).
	\end{equation}
\end{theorem}
\begin{proof}
	First, remark that any model of $B$ includes $\alpha$. Let us denote by $\wgtT^{A}, \wgtM^{A}, \wgtC^{A}$ and $\wgtE^{A}$ the weight functions for $A$, and likewise for $B$.
	
	Recall that the \acl{WF} $\weightfact{\alpha}{1.0}$ in $A$ is replaced by $\alpha \vee \co{\alpha}$ in the derived program $A'$. Therefore, some models, including the \aclp{SM}, of $A'$ include $\alpha$ while others include $\co{\alpha}$. On the other hand, any model of $B$ includes $\alpha$.
	The main idea of this proof is that for a \acl{TC} $t\in\TCHOICEset\at{A}$, \cref{eq:def.total.choice} entails that either $t = t_B \cup \set{\alpha}$ or $t = t_B \cup \set{\co{\alpha}}$ for some $t_B \in \TCHOICEset\at{B}$.
	

	\bigskip
	If $t = t_B \cup \set{\alpha}$ then
	$$
	\begin{aligned}
		\wgtT^A\at{t} &= \wgtT^A\at{t_B \cup \set{\alpha}} 
				& \text{[\cref{eq:def.total.choice}]} \\
			&= \prod_{\substack{
						\weightfact{a}{w}~ \in ~\WEIGHTFset\at{A} \\
						a~ \in~ t_B \vee a = \alpha}}
					w
				\quad \times\quad
				\prod_{\substack{
						\weightfact{a}{w}~ \in ~\WEIGHTFset\at{A}\\
						\co{a}~ \in~ t_B \vee \co{a} = \alpha}} 
					\co{w}		 
				& \text{[\cref{eq:weight.total.choice}]} \\
			&= \prod_{\substack{
						\weightfact{a}{w}~ \in ~\WEIGHTFset\at{A} \\
						a~ \in~ t_B}}
					w
				\quad \times\quad
				\prod_{\substack{
						\weightfact{a}{w}~ \in ~\WEIGHTFset\at{A}\\
						\co{a}~ \in~ t_B}} 
					\co{w} 
				& \text{[if}~ a = \alpha~\text{then}~w = 1.0 ~\text{and}~\forall a \del{\co{a} \not= \alpha}\text{]} \\
			&= \prod_{\substack{
						\weightfact{a}{w}~ \in ~\WEIGHTFset\at{B}, \\
						a~ \in~ t_B}}
					w
				\quad \times\quad
				\prod_{\substack{
						\weightfact{a}{w}~ \in ~\WEIGHTFset\at{B}\\
						\co{a}~ \in~ t_B}} 
					\co{w} 
				& {\text{[because}~\TCHOICEset\at{A} = \TCHOICEset\at{B \cup \set{\alpha}}\text{]}}\\
			&= \wgtT^B\at{t_B}. 
	\end{aligned}
	$$ 

	If $t = t_B \cup \set{\co{\alpha}}$ then 
	$$
	\begin{aligned}
		\wgtT^A\at{t} &= \wgtT^A\at{t_B \cup \set{\co{\alpha}}} 
				&& \text{[\cref{eq:def.total.choice}]} \\
			&= \prod_{\substack{
						\weightfact{a}{w}~ \in ~\WEIGHTFset\at{A} \\
						a~ \in~ t_B \vee a = \co{\alpha}}}
					w
				\quad \times\quad
				\prod_{\substack{
						\weightfact{a}{w}~ \in ~\WEIGHTFset\at{A}\\
						\co{a}~ \in~ t_B \vee \co{a} = \co{\alpha}}} 
					\co{w}		 
				&& \text{[\cref{eq:weight.total.choice}]} 
	\end{aligned}
	$$ 
	The last product of the formula is equal to $0$ because $\weightfact{\alpha}{1.0}$ entails $\co{w} = 0$ when $\co{a} = \co{\alpha}$.

	Therefore 
	$$
	\wgtT^A\at{t} = \begin{cases}
		\wgtT^B\at{t_B} &\text{if}~t = t_B \cup \set{\alpha}, \\
		0 &\text{if}~t = t_B \cup \set{\co{\alpha}} \\
	\end{cases} 
	$$
	for some $t_B\in \TCHOICEset\at{B}$.
	
	Also, a \acl{SM} $s \in \MODELset\at{A}$ either contains $\alpha$ or $\co{\alpha}$. In the former case, it is also a \ac{SM} of $B$. Therefore
	$$
	\wgtT^A\at{t}\wgtM^A\at{s,t} =
	\begin{cases}
		 \wgtT^B\at{t_B}\wgtM^B\at{s,t_B} 
			&\text{if}~t = t_B \cup \set{\alpha}~\text{and}~\alpha \in s, \\
		0 &\text{if}~t = t_B \cup \set{\co{\alpha}}~\text{or}~\co{\alpha}\in s
	\end{cases}	
	$$
	because, if $\co{\alpha} \in s$ and $t = t_B \cup \alpha$ then $s \not\in \MODELset^A\at{t}$ so $\wgtM^A\at{s,t} = 0$.

	\bigskip
	Concerning the relation between $\wgtC^A$ and $\wgtC^B$.
	Let $e\in \EVENTSset\at{A}$. Note that $\EVENTSset\at{A} = \EVENTSset\at{B}$ and then:
	$$
	\begin{aligned}
		\wgtC^A\at{\class{e}} 
		&= \sum_{t \in \TCHOICEset\at{A}} \del{\wgtT^A\at{t}\wgtC^A\at{\class{e}, t}} 
			& \text{[\cref{eq:weight.class.unconditional}]} \\
		&= \sum_{t \in \TCHOICEset\at{A}} \del{ 
			\wgtT^A\at{t}
			\sum_{s\in\stablecore{e}}\wgtM^A\at{s, t} } 
			& \text{[\cref{eq:weight.class.consistent}]} \\
		&= \sum_{t \in \TCHOICEset\at{A}} \del{ 
			\sum_{s\in\stablecore{e}}\wgtT^A\at{t}\wgtM^A\at{s, t} } \\
		&= \sum_{t_B \in \TCHOICEset\at{B}} \del{ 
			\sum_{s\in\stablecore{e},~ \alpha \in s}\wgtT^B\at{t_B}\wgtM^B\at{s, t_B} }
			& \text{[from}~\wgtT^A\at{t}\wgtM^A\at{s,t}~\text{above]} \\
		&= \sum_{t_B \in \TCHOICEset\at{B}} \del{ 
			\wgtT^B\at{t_B}\sum_{s\in\stablecore{e},~ \alpha \in s}\wgtM^B\at{s, t_B} }
			\\
		&= \sum_{t_B \in \TCHOICEset\at{B}} \del{ 
			\wgtT^B\at{t_B} \wgtC^B\at{\class{e}, t} }
			& \text{[\cref{eq:weight.class.consistent}]} \\
		&= \wgtC^B\at{\class{e}}
			& \text{[\cref{eq:weight.class.unconditional}]}
	\end{aligned}
	$$

	\bigskip
	Concerning the relation between $\wgtE^A$ and $\wgtE^B$. Let $e\in\EVENTSset$ such that $\#\class{e} > 0$. Then:

	$$
	\begin{aligned}
		\wgtE^A\at{e} &= \sum_{t \in \TCHOICEset\at{A}} \wgtT^A\at{t}\frac{\wgtC^A\at{\class{e}, t}}{\# \class{e}}
			& \text{[\cref{eq:weight.events,eq:weight.events.unconditional}]}
			\\
		&= \frac{1}{\# \class{e}}\sum_{t \in \TCHOICEset\at{A}}
			\wgtT^A\at{t}\wgtC^A\at{\class{e}, t}
			&
			\\
		&= \frac{1}{\# \class{e}} \wgtC^A\at{\class{e}}
			&\text{[\cref{eq:weight.class.unconditional}]}
			\\
		&= \frac{1}{\# \class{e}} \wgtC^B\at{\class{e}}
			&\text{[from}~\wgtC^A=\wgtC^B~\text{above]}
			\\
		&= \frac{1}{\# \class{e}} \sum_{t \in \TCHOICEset\at{B}}
		\wgtT^B\at{t}\wgtC^B\at{\class{e}, t}
			&\text{[\cref{eq:weight.class.unconditional}]}
			\\
		&= \sum_{t \in \TCHOICEset\at{B}} \wgtT^B\at{t}\frac{\wgtC^B\at{\class{e}, t}}{\# \class{e}}
			& 
			\\		
		&= \wgtE^B\at{e}
			& \text{[\cref{eq:weight.events,eq:weight.events.unconditional}]}
	\end{aligned}
	$$
\hfill
\end{proof}

\begin{corollary}[Probability of Certain Facts]
	\label{cor:prob.one}

	Consider a program $A$ with the \acl{WF} $\weightfact{\alpha}{1.0}$ and
	$B$ that results from $A$ by replacing that fact by the deterministic fact $\alpha$. Let $\prE^A$ given by \cref{eq:probability.event} for $A$ and $\prE^B$ for $B$.
	Then
	\begin{equation}
		\forall e \in \EVENTSset \left(\prE^A\at{e} = \prE^B\at{e}\right).
	\end{equation}
\end{corollary}
\begin{proof}	
	By \cref{teo:prob.one}, $\wgtE^A = \wgtE^{B}$ implies $Z^A = Z^{B}$, by \cref{eq:normalizing.factor}, and then $\prE^A = \prE^{B}$ by \cref{eq:probability.event}. 
\end{proof}

\begin{example}[Probability of Events]
	\label{ex:prob.events}
	\em
	The coherent \emph{prior} probability of events of program \SBF\ in \cref{ex:fruitful} is
	\begin{equation}
		\begin{array}{l|ccccccccc}
			\stablecore{e}          &
			\inconsistent           &
			\indepclass             &
			\co{a}                  &
			ab                      &
			ac                      &
			\co{a}, ab              &
			\co{a}, ac              &
			ab, ac                  &
			\consequenceclass
			\\ \hline
			\prE\at{e}              &
			0                       &
			0                       &
			\frac{7}{207}           &
			\frac{1}{23}\theta      &
			\frac{1}{23}\co{\theta} &
			0                       &
			0                       &
			\frac{3}{46}            &
			\frac{10}{23}
		\end{array}
		\label{eq:sbf.prior}
	\end{equation}

	This table can be used to compute the probability of any single event $e \in \EVENTSset$ by looking at the column of the event's \acl{SC}.

	For example:
	\begin{list}{\hspace{2em}}{\setlength\itemsep{0.5em}}
		\item $\wgtE\at{ab} = \frac{\theta}{23}, $ because $ab$ is the only \ac{SM} related with $ab$ so $\stablecore{ab} = \set{ab}$ and the weight value is found in the respective column of \cref{eq:sbf.prior}.

		\item $\wgtE\at{abc} = \frac{3}{46}$ because $abc \supset  ab$ and $abc \supset ac$.
		So $\stablecore{abc} = \set{ab, ac}$.

		\item $\wgtE\at{bc} = 0$ because, since there is no \ac{SM} $s$ that either $s \subset bc$ or $bc \subset s$, $\stablecore{bc} = \emptyset$ \emph{i.e.}\ $bc \in \indepclass$.

		\item $\wgtE\at{\co{a}b} = \frac{7}{207}$ because $\stablecore{\co{a}b} = \set{\co{a}}$.

		\item $\wgtE\at{\co{a}} = \frac{7}{207}$ and  $\wgtE\at{a} = \frac{3}{46}$.
		Notice that $\wgtE\at{\co{a}} + \wgtE\at{a} \not= 1$. 
		This highlights the fundamental difference between $\wgtE$ and $\wgtT$ (\emph{cf.~}\cref{prop:two.distributions}), where the former results from the lattice of the \aclp{SC} and the latter directly from the explicit assignment of probabilities to literals.
	\end{list}

	Related with the last case above, consider the complement of a consistent event $e$, denoted by $\complement e$.\defnote{$\complement e$}
	To calculate $\prE\at{\complement e}$ we look for the classes in $\class{\EVENTSset}$ that are not $\class{e}$, \ie~the complement of $e$'s class within $\class{\EVENTSset}$\footnote{All the usual set operations hold on the complement.
	For example, $\complement\complement X = X$.}, $\complement\class{e}$.
	Considering that $\class{\EVENTSset}$ is in a one-to-one correspondence with the \aclp{SC} plus $\inconsistent$, 
	\begin{equation*}
		\class{\EVENTSset} \simeq \set{
			\inconsistent, \indepclass, \set{\co{a}}
			, \set{ab}, \set{ac}, \set{\co{a}, ab}
			, \set{\co{a}, ac}, \set{ab, ac}, \consequenceclass}.
	\end{equation*}
	In particular, for $\wgtE\at{\complement a}$, since $\stablecore{a} = \set{ab, ac}$ then $\complement \class{a} = \class{\EVENTSset} \setminus \class{a}$ and
	\(
	\prE\at{\complement a} =  \prE\at{\class{\EVENTSset} \setminus \class{a}} = 1 - \prE\at{a}
	\).
	Also, $\prE\at{\complement \co{a}} =  1 - \prE\at{\co{a}} $.
\end{example}

While not illustrated in our examples, this method also applies to programs that have more than one \acl{PF}, like
\begin{equation*}
	\left\{
	\begin{aligned}
		\weightfact{a & }{0.3},             %
		&&
		\weightfact{b & }{0.6},             %
		\\
		c \vee d    & \clause a \wedge b.
	\end{aligned}
	\right.
\end{equation*}

Our approach generalizes to Bayesian networks in a way similar to
\cite{cozman2020joy,raedt2016statistical} and
\cite{kiessling1992database,thone1997increased} as follows.
On the one hand, any acyclic propositional program can be viewed as the specification of a Bayesian network over binary random variables.
So, we may take the structure of the Bayesian network to be the dependency graph.
The random variables then correspond to the atoms and the probabilities can be read off of the probabilistic facts and rules.
Conversely, any Bayesian network over binary variables can be specified by an acyclic non-disjunctive \ac{WASP}.

\section{Discussion and Future Work}
\label{sec:discussion}

This work is a first venture into expressing weight assignments using algebraic expressions derived from a logical program, in particular an \ac{ASP}.
We would like to point out that there is still much to explore concerning the full expressive power of logic programs and \ac{ASP} programs.
So far, we have not considered recursion, logical variables or functional symbols.

The theory, methodology, and tools, from Bayesian Networks can be adapted to our approach.
The connection with Markov Fields \cite{kindermann80} is left for future work.
An example of a `program selection' application (as mentioned in \cref{item:program.selection}, \cref{ssec:propagating.weights}) is left for future work.
Also, there is still little effort concerning the \emph{Probabilistic Tasks} and to articulate with the related fields of probabilistic logical programming, machine learning, inductive programming, \emph{etc.}

The equivalence relation from \cref{def:equiv.rel} identifies the $s \subseteq e$ and $e \subseteq s$ cases.
Relations that distinguish such cases might enable better relations between the representations and processes from the \aclp{SM}.

Furthermore, we decided to set the weight of inconsistent events to $0$ but, maybe, in some cases, we shouldn't.
For example, since observations may be affected by noise, one can expect to occur some inconsistencies.

\section*{Acknowledgements}

This work is partly supported by Funda\c{c}\~ao para a Ci\^{e}ncia e Tecnologia (FCT/IP) under contracts UIDB/04516/2020 (NOVA LINCS), UIDP/04674/2020 and UIDB/04675/2020 (CIMA).

The authors are grateful to Lígia Henriques-Rodrigues, Matthias Knorr and Ricardo Gonçalves for valuable comments on a preliminary version of this paper, and Alice Martins for contributions on software development.

\bibliographystyle{plain}
\bibliography{zugzwang}

\end{document}

%% file: macros.tex
\usepackage{multirow}

\usepackage{times}
\usepackage{url}
\usepackage[hidelinks]{hyperref}
\usepackage{graphicx}
\usepackage{booktabs}
\usepackage{algorithm}
\usepackage{algorithmic}
\usepackage[
    prependcaption]{todonotes}
\usepackage{tikz}
\tikzset{
    boxed/.style={draw, rectangle, rounded corners=4pt, fill=gray!10,inner sep=1em, anchor=center},
    event/.style={},
    smodel/.style={fill=gray!25},
    tchoice/.style={draw, circle},
    indep/.style={},
    proptc/.style = {-latex, dashed},
    propsm/.style = {-latex, thick},
    doubt/.style = {gray} }
\usetikzlibrary{calc, positioning, patterns, perspective}
\usepackage{hyperref}
\hypersetup{
    colorlinks=true,
    linkcolor=blue,
    citecolor=blue,
    urlcolor=blue, }
\usepackage{commath}

\newtheorem{assumption}{Assumption}
\usepackage{amssymb}
\usepackage[normalem]{ulem}
\usepackage[nice]{nicefrac}
\usepackage{stmaryrd}
\usepackage[smaller]{acronym}
\usepackage{multicol}
\usepackage{multirow}
\usepackage{cleveref}
\crefname{example}{ex.}{exs.}
\crefname{proposition}{prop.}{props.}
\crefname{assumption}{assumption}{assumptions}
\usepackage{hyphenat}
\newtheorem{example}{Example}
\newtheorem{definition}{Definition}
\newtheorem{proposition}{Proposition}
\newtheorem{corollary}{Corollary}

\def\tight{%
  \itemsep 0pt plus 1pt
  \parskip 0pt plus 1pt}

\newcommand{\ie}{\emph{i.e.}}
\newcommand{\eg}{\emph{e.g.}}
\newcommand{\eat}[1]{}
\newcommand{\at}[1]{\ensuremath{\!\del{#1}}}        
\newcommand{\cla}[1]{\ensuremath{{\mathcal{#1}}}}        

\newcommand{\conj}{\ensuremath{\wedge}} \newcommand{\disj}{\ensuremath{\vee}}
\newcommand{\clause}{\ensuremath{\leftarrow}}
\DeclareMathOperator{\naf}{\sim\!}
\newcommand{\co}[1]{\ensuremath{\overline{#1}}}     
\newcommand{\ATOMSset}{\ensuremath{\cla{A}}}

\newcommand{\LITERALSset}{\ensuremath{\cla{L}}}

\newcommand{\WATOMset}{\ensuremath{\ATOMSset_{\cla{W}}}}
\newcommand{\FACTSset}{\ensuremath{\cla{F}}}

\newcommand{\WEIGHTFset}{\ensuremath{\cla{W}}}
\newcommand{\RULESset}{\ensuremath{\cla{R}}}
\newcommand{\TCHOICEset}{\ensuremath{\cla{T}}}
\newcommand{\MODELset}{\ensuremath{\cla{S}}} 
\newcommand{\EVENTSset}{\ensuremath{\cla{E}}}
\newcommand{\CONSISTset}{\ensuremath{\cla{C}}}
\newcommand{\SBF}{\ensuremath{P_{\sbf}}}
\newcommand{\prfunc}{\ensuremath{\mathrm{P}}}       
\newcommand{\pr}[1]{\ensuremath{\prfunc\at{#1}}}    
\newcommand{\prd}[1]{\ensuremath{\prfunc_{#1}}}     
\newcommand{\prT}{\prd{\TCHOICEset}}

\newcommand{\prE}{\prd{\EVENTSset}}

\newcommand{\wgtfunc}{\ensuremath{\omega}}       
    
\newcommand{\wgtd}[1]{\ensuremath{\wgtfunc_{#1}}}     
\newcommand{\wgtT}{\wgtd{\TCHOICEset}}
\newcommand{\wgtM}{\wgtd{\MODELset}}
\newcommand{\wgtE}{\wgtd{\EVENTSset}}
\newcommand{\wgtC}{\wgtd{\cla{R}}}
\newcommand{\wgte}[1]{\ensuremath{\wgtE\at{#1}}}
\newcommand{\wgtm}[1]{\ensuremath{\wgtM\at{#1}}}
\newcommand{\wgtt}[1]{\ensuremath{\wgtT\at{#1}}}
\newcommand{\wgtc}[1]{\ensuremath{\wgtC\at{#1}}}
\newcommand{\pwT}{\ensuremath{\mu_{\TCHOICEset}}}   
\newcommand{\pwt}[1]{\ensuremath{\pwT\at{#1}}}     
\newcommand{\pwM}{\ensuremath{\mu_{\MODELset}}}   
     
\newcommand{\pwE}{\ensuremath{\mu_{\EVENTSset}}}   
     
\newcommand{\stablecore}[1]{\ensuremath{\left\llbracket #1 \right\rrbracket}}
\newcommand{\inconsistent}{\bot}
\newcommand{\given}{\ensuremath{~\middle|~}}
\newcommand{\consequenceclass}{\ensuremath{\Lambda}}
\newcommand{\indepclass}{\ensuremath{\Diamond}}

\newcommand{\weightfact}[2]{\ensuremath{#1:#2}}
\newcommand{\weightrule}[3]{\weightfact{#1}{#2} \leftarrow #3}
\newcommand{\class}[1]{\ensuremath{[{#1}]_{\sim}}}
\newcommand{\tcgen}[1]{\MODELset\at{#1}}

\newcommand{\lpmln}{\texttt{LP\textsuperscript{MLN}}}
\newcommand{\emptyevent}{\ensuremath{\lambda}}
\newcommand{\powerset}[1]{\ensuremath{\mathrm{2}^{#1}}}
\acrodef{BK}[BK]{background knowledge}
\acrodef{ASP}[ASP]{answer set program}
\acrodef{NP}[NP]{normal program}
\acrodef{DP}[DP]{disjunctive program}
\acrodef{LP}[LP]{logic program}
\acrodef{PCR}[PCR]{program with choice rules}
\acrodef{DS}[DS]{distribution semantics}
\acrodef{PF}[PF]{probabilistic fact}
\acrodef{WF}[WF]{weighted fact}
\acrodef{TC}[TC]{total choice}
\acrodef{SM}[SM]{stable model}
\acrodef{SC}[SC]{stable core}
\acrodef{KL}[KL]{Kullback-Leibler}
\acrodef{SBF}[SBF]{simple but fruitful}
\acrodef{RSL}[RSL]{random set of literals}
\acrodef{RCE}[RCE]{random consistent event}
\acrodef{SASP}[SASP]{stochastic answer set program}
\acrodef{WASP}[WASP]{weighted answer set program}
\acrodef{HMM}[HMM]{hidden Markov model}
\acrodef{MAP}[MAP]{maximum a posteriori}
\acrodef{MLE}[MLE]{maximum likelihood estimation}
\acrodef{BI}[BI]{Bayesian inference}
\acrodef{PLP}[PLP]{probabilistic logic programming}
\newcommand{\sbf}{\ensuremath{\mathrm{1}}}
\newcounter{revcounter}
\ifREVEDIT
    
    \newcommand{\delete}[1]{\sout{#1}}
    \newcommand{\sidenote}[1]{\stepcounter{revcounter}{\color{red!50!black}\(\vert^{\arabic{revcounter}}\)}\marginpar{{\color{red!50!black}\(^{\arabic{revcounter}}\vert\)}\scriptsize #1}}
    \newcommand{\defnote}[1]{\marginpar{\scriptsize{{\color{blue!50!black}\bf def.~}#1}}}
    \newcommand{\topicnote}[1]{{\scriptsize\color{red!50!black}$\blacktriangleright$}\marginpar{\scriptsize{\color{red!50!black}\it #1}}}
    \newcommand{\replace}[2]{\delete{#1}\sidenote{#2}}
    \newcommand{\franc}[1]{{\color{green!30!black}#1}}

    \newcommand{\dietmar}[1]{{\color{brown!40!black}#1}}

    \newcommand{\selfnote}[1]{\todo[backgroundcolor=green!20]{{\footnotesize #1}}}
    \newcommand{\spanote}[1]{{\todo[size=footnotesize,color=teal!20]{\textbf{SPA:} #1}}}
    \newcommand{\dsnote}[1]{{\todo[size=footnotesize,color=teal!20]{\textbf{DS:} #1}}}
    \newcommand{\francnote}[1]{{\todo[size=footnotesize,color=green!30]{\textbf{FC:} #1}}}
    \newcommand{\bdnote}[1]{{\todo[size=footnotesize,color=red!60]{\textbf{BD:} #1}}}
\else
    
    \newcommand{\delete}[1]{}
    \newcommand{\sidenote}[1]{}
    \newcommand{\defnote}[1]{}
    \newcommand{\topicnote}[1]{}
    \newcommand{\replace}[2]{}
    \newcommand{\franc}[1]{}

    \newcommand{\dietmar}[1]{}

    \newcommand{\selfnote}[1]{}
    \newcommand{\spanote}[1]{}
    \newcommand{\dsnote}[1]{}
    \newcommand{\francnote}[1]{}
    \newcommand{\bdnote}[1]{}
\fi